\documentclass[11pt]{article}
\usepackage[margin=1in]{geometry}
\usepackage{amsmath,amssymb,amsthm,mathtools}
\usepackage{tikz}
\usepackage[nopatch=footnote]{microtype}
\usepackage[hidelinks]{hyperref}
\newtheorem{theorem}{Theorem}[section]
\newtheorem{proposition}[theorem]{Proposition}
\newtheorem{lemma}[theorem]{Lemma}

\theoremstyle{definition}
\newtheorem{definition}[theorem]{Definition}
\theoremstyle{remark}

\newcommand{\Sym}{\mathfrak S}
\newcommand{\Sepp}{\operatorname{Sepp}}
\newcommand{\ord}{\operatorname{ord}}
\newcommand{\lcm}{\operatorname{lcm}}
\newcommand{\cP}{\mathcal P}
\newcommand{\cF}{\mathcal F}
\title{An Improved Lower Bound for Separating Words\\
by Permutation Automata}
\author{Chen Xu and Kenneth Regan}
\date{October 8, 2026}
\begin{document}
\maketitle

\begin{abstract}
We prove that some pairs of binary words of every sufficiently large
prescribed length $n$ require $\Omega(\log n\log\log n)$ states to separate
by a deterministic finite automaton whose letters act as permutations.
We obtain this bound by constructing distinct equal-length positive binary
words that form an identity of $\Sym_k$ of length $\exp(O(k/\log k))$.
This improves the previous lower bound $(3/2-o(1))\log n$ of~\cite{BKSS}.
The proof combines a positive-word construction with a cover
of permutation orders.  Consecutive segments of prime powers reduce the
covering problem to enumerating sets of positive integers with bounded sum.
The required estimates use the central binomial coefficient and the
classical asymptotic formula for partitions into distinct parts.  An appendix gives sharper constants and
lower-order terms using a greedy cover and a
quantitative prime number theorem.
\end{abstract}

\section{Introduction}
\label{sec:introduction}

A deterministic finite automaton \emph{separates} two words when it accepts
exactly one of them.  In a \emph{permutation automaton}, each input letter
acts bijectively on the state set.  For distinct binary words $u,v$, let
$\Sepp(u,v)$ be the minimum number of states of such an automaton separating
them.  Define
\[
 \Sepp(n)=\max_{\substack{u\ne v\\|u|,|v|\le n}}\Sepp(u,v),
 \qquad
 \Sepp_{=}(n)=\max_{\substack{u\ne v\\|u|=|v|=n}}\Sepp(u,v).
\]
The alphabet is $\{0,1\}$, with letter permutations $\sigma_0,\sigma_1$.
All logarithms are natural.  Words act on the right:
$q\cdot uv=(q\cdot u)\cdot v$.

Robson~\cite{Robson} studied separation by permutation automata and proved
an $O(\sqrt n)$ upper bound.  Bulatov, Karpova, Shur, and
Startsev constructed positive
identities of $\Sym_k$ of length $\exp(2k/3+O(k/\log k))$, giving
$\Sepp(n)\ge(3/2-o(1))\log n$
(Theorem~16 and Corollary~17 in~\cite{BKSS}).  For $k\ge5$, Karpova and
Shur characterized identities of the form
$(01)^a(10)^b=(10)^b(01)^a$ by two-integer covers and proved this length
bound optimal within that family (Theorems~18 and~21 in~\cite{KS}).
Our main result improves the order of the lower bound.

\begin{theorem}\label{thm:elementary-sepp}
$\Sepp(n)\ge\Sepp_{=}(n)=\Omega(\log n\log\log n)$.
\end{theorem}

\noindent
The constant in the $\Omega$ is the reciprocal of $12\sqrt{2}$ in the proof in Section~\ref{sec:elementary-cover}.
Appendix~\ref{sec:refinement} sharpens the constant to $1$, giving 
$\Sepp_{=}(n)\ge\log n\log\log n\cdot(1-o(1))$, and gives explicit lower-order
terms, using a quantitative prime number theorem.

A \emph{positive binary identity} of $\Sym_k$ is a pair of distinct words
$u,v$ over $\{0,1\}$ satisfying
$u(\sigma_0,\sigma_1)=v(\sigma_0,\sigma_1)$ for every
$\sigma_0,\sigma_1\in\Sym_k$.  It is \emph{homogeneous} if $|u|=|v|$.  Write
$\ell_{+,2}(k)$ for the minimum common length of a homogeneous positive
binary identity.

The identity $u=v$ holds in $\Sym_k$ if and only if no permutation
automaton with at most $k$ states separates $u,v$; this correspondence
is discussed in Section~8 of~\cite{DESW}.  Embed the letter
permutations in $\Sym_k$ by adding fixed points; an identity then gives
the same endpoint from every state.  Conversely, if the evaluations differ,
start at a state with different images and accept only one image.

We prove the automata theorem by constructing short identities.

\begin{theorem}\label{thm:elementary-identity}
For every sufficiently large $k$, there are distinct positive binary
words $u_k,v_k$ with equal counts of each letter, forming an identity
of $\Sym_k$, with
\[
 |u_k|=|v_k|=\exp(O(k/\log k)).
\]
\end{theorem}

For all sufficiently large $n$, Theorem~\ref{thm:elementary-identity}
gives an identity of length at most $n$ for some
$k=\Theta(\log n\log\log n)$.  The distinct words
$u_k0^{n-|u_k|}$ and $v_k0^{n-|u_k|}$ have length exactly $n$ and
form an identity of $\Sym_k$.  The correspondence above therefore
gives $\Sepp(n)\ge\Sepp_{=}(n)>k$, proving
Theorem~\ref{thm:elementary-sepp}.

The proof extends the two-integer covering approach of Bulatov
et al.\ (Proposition~15 in~\cite{BKSS}) to arbitrary finite covers, using a
recursion of positive words.  We build the cover by partitioning prime powers
into consecutive segments of controlled logarithmic weight.  This reduces
the covering problem to counting sets of positive integers with bounded
sum and gives at most $\sqrt{k}$ cover integers, each at most
$\exp(O(k/\log k))$ in magnitude.  Combining these integers into an identity contributes
only $O(\sqrt{k})$ to its logarithmic length, so the largest cover integer
determines the asymptotic scale.

After the preliminary estimates, Sections~\ref{sec:selector}
and~\ref{sec:elementary-cover} prove
Theorem~\ref{thm:elementary-identity}.  Section~\ref{sec:tightness} proves
that its order of magnitude is optimal for this construction, and Section~\ref{sec:disclosure} documents AI use including efforts to push the lower bound further by other ideas.

\section{Preliminaries}

We begin with three arithmetic estimates used in the construction.

\begin{lemma}[Subset count]\label{lem:elementary-count}
As $Q\to\infty$, the number of sets of positive integers that sum to
at most $Q$ is asymptotic to
\[
 \frac{3^{1/4}}{2\pi Q^{1/4}}
 \exp\!\left(\pi\sqrt{\frac Q3}\right).
\]
\end{lemma}
\begin{proof}
Let $p_{\mathrm d}(m)$ count partitions of $m$ into distinct parts,
with $p_{\mathrm d}(0)=1$, and put
$Q(n)=\sum_{m=0}^n p_{\mathrm d}(m)$.
Theorem~1.1 in~\cite{LV24} gives
\[
 p_{\mathrm d}(n)\sim
 \frac{\exp(\pi\sqrt{n/3})}{4\cdot 3^{1/4}n^{3/4}}.
\]
Sum--integral comparison, followed by $t=s^2$ and integration by parts,
gives
\[
 \begin{aligned}
 Q(n)&\sim\frac1{4\cdot 3^{1/4}}
       \int_1^n t^{-3/4}\exp(\pi\sqrt{t/3})\,dt\\
     &\sim\frac{3^{1/4}}{2\pi n^{1/4}}
       \exp(\pi\sqrt{n/3}).
 \end{aligned}
\]
\end{proof}

We next bound the logarithmic weight of prime powers.

\begin{lemma}[Weight bound]\label{lem:elementary-weight}
For any $x\ge0$, summing over primes $p$,
\[
 A(x):=\sum_{p^a\le x}\log(p^a)\le3x.
\]
\end{lemma}
\begin{proof}
Put $\psi(x)=\sum_{p^a\le x}\log p$.  The prime exponents in the
central binomial coefficients give, for every integer $m\ge1$,
\[
 \begin{aligned}
 \psi(2m)-\psi(m)
 &\le\log\binom{2m}{m}\le2m\log2,\\
 \psi(2m+1)-\psi(m+1)
 &\le\log\binom{2m+1}{m}\le2m\log2.
 \end{aligned}
\]
For the last inequality, the two central terms in the binomial expansion
of $2^{2m+1}$ are equal.  Induction from $\psi(1)=0$ gives
$\psi(n)\le2n\log2$ for all integers $n\ge1$, hence
$\psi(x)\le2x\log2$ for every real $x\ge0$.

For $1\le b\le a$, set $j=a-b$.  Since $p\ge2$,
$2^jp^b\le p^{b+j}=p^a$.  Therefore
\[
 \begin{aligned}
 A(x)&=\sum_{p^a\le x}\sum_{b=1}^a\log p
       \le\sum_{j\ge0}\psi(x/2^j)\\
     &\le2x\log2\sum_{j\ge0}2^{-j}
       =4x\log2\le3x.
 \end{aligned}
\]
\end{proof}

Let $q=p_1^{a_1}\cdots p_t^{a_t}$ be the unique prime factorization of
$q$, with distinct primes $p_i$ and positive integers $a_i$.  Define
\[
 S(q)=\{p_1^{a_1},\ldots,p_t^{a_t}\}.
\]
Thus $S(q)$ is the set of maximal prime-power divisors of $q$;
for instance, $S(12)=\{3,4\}$.  For $q=1$, $S(1)=\varnothing$.
Empty sums and products equal zero and one, respectively.
The following bound is due to Erd\H{o}s and
Tur\'an (Section~7 in~\cite{ErdosTuran}).

\begin{lemma}[Support cost]\label{lem:support}
If $q$ is an element order of $\Sym_k$, then
\[
 q=\prod_{s\in S(q)}s,\qquad \sum_{s\in S(q)}s\le k.
\]
\end{lemma}
\begin{proof}
Let $c_1,\ldots,c_t$ be the cycle lengths of a permutation of order $q$.
Then
\[
 q=\lcm(c_1,\ldots,c_t)\qquad\text{and}\qquad
 \sum_{i=1}^t c_i=k.
\]
Thus $S(q)\subseteq\bigcup_{i=1}^t S(c_i)$.
Repeatedly applying $a+b\le ab$ for $a,b\ge2$ gives
\[
 \sum_{s\in S(c_i)}s\le\prod_{s\in S(c_i)}s=c_i
 \qquad(1\le i\le t),
\]
including $c_i=1$ under the empty-set conventions. Therefore
\[
 \sum_{s\in S(q)}s
 \le\sum_{i=1}^t\sum_{s\in S(c_i)}s
 \le\sum_{i=1}^t c_i=k.
\]
Finally, prime factorization gives
$q=\prod_{s\in S(q)}s$.
\end{proof}

\section{Positive identities and covers}
\label{sec:selector}

We first record two reductions for constructing identities.
Assume $u$ and $v$ are distinct equal-length words.

\begin{lemma}\label{tricks}
\begin{itemize}
\item[(a)]
Without loss of generality, $u$ and $v$ begin and end with different letters.
\item[(b)]
Without loss of generality (up to constant factors), we can restrict attention to pairs $\sigma_0,\sigma_1$ of the same order.
\end{itemize}
\end{lemma}

\begin{proof}
(a)~If they begin with the same letter or end with the same letter, then deleting the letter leaves shorter strings that also give a global identity.
(b)~Define $u'$ and $v'$ as the images of $u$ and $v$ under the length-doubling homomorphism $0 \mapsto 01$, $1 \mapsto 10$.  Then 
\[
u(xy,yx) = u'(x,y) \qquad\text{and}\qquad v(xy,yx) = v'(x,y).
\]
Thus $u' = v'$ is an identity if and only if $u(\sigma_0,\sigma_1) = v(\sigma_0,\sigma_1)$ for all cases of the form $\sigma_0 = xy$, $\sigma_1 = yx$.  Since $xy$ and $yx$ always have the same order, (b) follows.
\end{proof}

\medskip
Call $u=v$ in (b) a {\bf same-order identity}.  Our constructions will avail the property that both $\ord(\sigma_0)$ and $\ord(\sigma_1)$ divide the same $d_i$ when one of them does.  This motivates the key concept:

\medskip
\begin{definition}\label{def:cover}
A \emph{cover} of $\Sym_k$ is a finite set $\mathcal C_k$ of positive
integers such that every element order of $\Sym_k$ divides some
$d\in\mathcal C_k$.
\end{definition}

\noindent
An integer in a cover need not itself be a permutation order.  Now we turn to strings from which we will compose group identities.  A \emph{block} of a word is a maximal consecutive sequence of equal letters.

\begin{definition}
For a positive binary word $W$ and an integer $d\ge1$, define its
\emph{positive $d$-inverse} $W^{\dagger d}$ by reversing the order of
its blocks and replacing each exponent $a$ by the unique integer
in $\{1,\ldots,d\}$ congruent to $-a$ modulo $d$.
\end{definition}

For example, the positive $5$-inverse of $0^7 1^5$ is $1^5 0^3$.
When $\sigma_0^d=\sigma_1^d=1$, this word evaluates to
$W(\sigma_0,\sigma_1)^{-1}$.
The replacements are positive, including $d$ for multiples of $d$,
so the number of blocks is unchanged and the new length is at most $d$ times
that count.

Throughout this section, write $r=|\mathcal C_k|$ for the size
of the cover.

\begin{lemma}[Word construction]\label{lem:selector}
Let $\mathcal C_k$ be a cover.
There are distinct positive binary words $U_r,V_r$ such that, for
$\sigma_0,\sigma_1\in\Sym_k$,
\[
 \sigma_0^d=\sigma_1^d=1\text{ for some }d\in\mathcal C_k
 \quad\Longrightarrow\quad
 U_r(\sigma_0,\sigma_1)=V_r(\sigma_0,\sigma_1).
\]
\end{lemma}

\begin{proof}
Enumerate $\mathcal C_k=\{d_1,\ldots,d_r\}$, not necessarily in increasing order.
Start with $U_1=0^{d_1}$ and $V_1=1^{d_1}$.  For $1\le i<r$, define
\begin{equation}\label{eq:selector-recursion}
 C_i=U_i^{\dagger d_{i+1}},\qquad
 U_{i+1}=U_iC_iV_i,\qquad V_{i+1}=V_iC_iU_i.
\end{equation}
Fix $\sigma_0,\sigma_1\in\Sym_k$ and evaluate the words by
$0\mapsto\sigma_0$, $1\mapsto\sigma_1$.
If $\sigma_0^{d_1}=\sigma_1^{d_1}=1$, then $U_1=V_1=1$.
If $U_i=V_i$, then
\[
 U_{i+1}=U_iC_iV_i=U_iC_iU_i=V_iC_iU_i=V_{i+1}.
\]
If $\sigma_0^{d_{i+1}}=\sigma_1^{d_{i+1}}=1$, the positive inverse gives
$C_i=U_i^{-1}$, so
\[
 \begin{aligned}
 U_{i+1}&=U_iU_i^{-1}V_i=V_i,\\
 V_{i+1}&=V_iU_i^{-1}U_i=V_i.
 \end{aligned}
\]
Induction proves the stated implication.

As formal words, $U_i$ and $V_i$ begin with $0$ and $1$, respectively,
so they are distinct.
\end{proof}

\begin{lemma}[Length bound]\label{lem:selector-length}
Let $U_i,V_i$ be the words constructed in Lemma~\ref{lem:selector}.
For $1\le i\le r$, each word has at most
$3^{i-1}$ blocks, and
\[
 |U_i|=|V_i|\le3^{i-1}\max\mathcal C_k.
\]
\end{lemma}
\begin{proof}
We induct on $i$.  The case $i=1$ follows from
$U_1=0^{d_1}$ and $V_1=1^{d_1}$.  At the next step, $C_i$ has at most
$3^{i-1}$ blocks, each of length at most
$d_{i+1}\le\max\mathcal C_k$.  Hence
\[
 \begin{aligned}
 |C_i|&\le3^{i-1}\max\mathcal C_k,\\
 |U_{i+1}|=|V_{i+1}|
 &=|U_i|+|C_i|+|V_i|\le3^i\max\mathcal C_k.
 \end{aligned}
\]
Each concatenation in~\eqref{eq:selector-recursion} has at most
$3\cdot3^{i-1}=3^i$ blocks, completing the induction.
\end{proof}

\begin{lemma}[Binary identity]\label{lem:cover-identity}
Every cover $\mathcal C_k$ gives a positive binary identity $u=v$
of $\Sym_k$ with equal counts of each letter and
\begin{equation}\label{eq:selector-upper}
 \log|u|\le\log2+\bigl(|\mathcal C_k|-1\bigr)\log3
             +\max_{d\in\mathcal C_k}\log d.
\end{equation}
\end{lemma}
\begin{proof}
Take the words $U_r,V_r$ constructed in Lemma~\ref{lem:selector}
for this cover.  Following the conjugate substitution from
Proposition~15 in~\cite{BKSS}, set $u=U_r(01,10)$ and $v=V_r(01,10)$.

For arbitrary $\sigma_0,\sigma_1\in\Sym_k$,
\[
 \sigma_1\sigma_0=\sigma_0^{-1}(\sigma_0\sigma_1)\sigma_0.
\]
Thus $\sigma_0\sigma_1$ and $\sigma_1\sigma_0$ have the same order,
which divides some $d\in\mathcal C_k$.  Lemma~\ref{lem:selector} gives
\[
 \begin{aligned}
 u(\sigma_0,\sigma_1)
 &=U_r(\sigma_0\sigma_1,\sigma_1\sigma_0)\\
 &=V_r(\sigma_0\sigma_1,\sigma_1\sigma_0)=v(\sigma_0,\sigma_1).
 \end{aligned}
\]
The words $u,v$ begin with $0,1$, respectively.  Each substitution
contributes one of each letter.  By Lemma~\ref{lem:selector-length},
\[
 |u|=|v|=2|U_r|\le2\cdot3^{r-1}\max\mathcal C_k,
\]
which gives~\eqref{eq:selector-upper}.
\end{proof}

These identities belong to Karpova and Shur's
\emph{numerical-2} class (Section~3.3 in~\cite{KS}).
For each possible common order $q$ of $01$ and $10$, some cover integer
is divisible by $q$.  The positive-inverse cancellations then reduce
both sides to the same word using only deletions of $(01)^q$ and $(10)^q$.

It remains to choose the integers $d_1,\ldots,d_r$ used in the recursion
for $U_r,V_r$.  We choose them so that every element order of $\Sym_k$
divides at least one $d_i$.  Lemma~\ref{lem:cover-identity} then ensures
that $u=U_r(01,10)$ and $v=V_r(01,10)$ form an identity.  The bound
$|u|=|v|\le2\cdot3^{r-1}\max_i d_i$ shows why we seek both a small
number of integers and a small maximum.

We form these integers as products of selected prime powers.
One product can cover several orders: for example,
\[
 12=3\cdot4\mid60,\qquad 15=3\cdot5\mid60,
 \qquad60=3\cdot4\cdot5.
\]
In general, if a selected set contains $S(q)$, which we defined as the set of maximal prime-power divisors
of $q$, its product is divisible by $q$.  Now define
\[
\cP_k = \{p^a: p^a \leq k\},
\]
where $p$ is prime and $a \geq 1$. By Lemma~\ref{lem:support}, every $S(q)$ arising from an element order
of $\Sym_k$ belongs to the finite family
\[
 \cF_k=\left\{S\subseteq\cP_k:\sum_{s\in S}s\le k\right\}
\]
(which also allows different powers of the same prime).
Choose subsets $T\subseteq\cP_k$ so that every $S\in\cF_k$ is contained
in at least one chosen $T$.  For each such $T$, include
$d=\prod_{z\in T}z$ in $\mathcal C_k$.  Then
\[
 S(q)\subseteq T
 \quad\Longrightarrow\quad
 q=\prod_{s\in S(q)}s\mid\prod_{z\in T}z=d,
\]
so these integers form a cover.  The number of chosen sets bounds the
cover size, and $\log d=\sum_{z\in T}\log z$.  Thus we need few sets $T$
with small logarithmic sums.  The sets $T$ need not themselves satisfy
$\sum_{z\in T}z\le k$.

\section{Constructing \texorpdfstring{$\mathcal C_k$}{Ck}}
\label{sec:elementary-cover}

A prime power contributes its value to the support cost and its logarithm
to $\log d$ for $d\in\mathcal C_k$.
We relate these quantities by partitioning prime powers
into segments in numerical order with approximately equal logarithmic weights.
Lemma~\ref{lem:elementary-weight} controls segment weights, while
Lemma~\ref{lem:elementary-count} bounds the number of possible occupied
index sets.

Fix integers $k\ge2$ and $Q\ge1$, and put $W=3k/Q$.  We construct a set
$\mathcal C_k$ as follows.
Partition $\cP_k$ into nonempty segments $B_0,B_1,\ldots$, with
$\max B_i<\min B_{i+1}$ for adjacent segments.  Then for every segment, the sum excluding the segment's last element satisfies
\[
\sum_{\substack{z\in B_i\\z<\max B_i}}\log z
 <W,
 \]
 while including the last element gives
 \[
 \sum_{z\in B_i}\log z
 \ge W
\]
except for the last segment.
Thus each segment ends when its sum first reaches $W$, or when no
elements remain.

For every set $J$ of positive segment indices with $\sum_{i\in J}i\le Q$,
including the empty set, form
\begin{equation}\label{eq:cover-integer}
 d_J=\left(\prod_{z\in B_0}z\right)
             \prod_{i\in J}\left(\prod_{z\in B_i}z\right).
\end{equation}
Let $\mathcal C_k$ be the set of these integers $d_J$.
Every product includes $B_0$ and all elements of each selected segment.
Figure~\ref{fig:cover-construction} illustrates one such product.

\begin{figure}[ht]
\centering
\begin{tikzpicture}[x=1cm,y=1cm,
 every node/.style={font=\small},
 segment/.style={draw=black!35,rounded corners=2pt,
                 minimum width=2.05cm,minimum height=0.75cm},
 selected/.style={segment,draw=blue!55!black,fill=blue!8}]
 \node at (5,1.35) {$k=Q=16,\qquad W=3$};
 \node at (0,0.7) {$B_0$};
 \node at (2.5,0.7) {$B_1$};
 \node at (5,0.7) {$B_2$};
 \node at (7.5,0.7) {$B_3$};
 \node at (10,0.7) {$B_4$};
 \node[segment,draw=black!60,fill=black!10] (b0) at (0,0) {$2,\ 3,\ 4$};
 \node[selected] (b1) at (2.5,0) {$\mathbf{5},\ 7$};
 \node[segment] (b2) at (5,0) {$8,\ 9$};
 \node[selected] (b3) at (7.5,0) {$\mathbf{11},\ 13$};
 \node[segment] (b4) at (10,0) {$16$};
 \node at (5,-0.85) {$J=\{1,3\},\quad 1+3\le16$};
 \node (p0) at (0,-1.8) {$(2\cdot3\cdot4)$};
 \node (p1) at (2.5,-1.8) {$(5\cdot7)$};
 \node (p3) at (7.5,-1.8) {$(11\cdot13)$};
 \node[anchor=east] at (-1.0,-1.8) {$d_{\{1,3\}}=$};
 \node at (1.25,-1.8) {$\times$};
 \node at (5,-1.8) {$\times$};
 \node[anchor=west] at (8.5,-1.8) {$=120120$};
 \draw[->,thick,black!65] (b0.south) -- (p0.north);
 \draw[->,thick,blue!55!black] (b1.south) -- (p1.north);
 \draw[->,thick,blue!55!black] (b3.south) -- (p3.north);
\end{tikzpicture}
\caption{The bold entries form $S(55)=\{5,11\}$, with $5+11=16$.
The shaded segments give $d_{\{1,3\}}=120120$, which is divisible by $55$.
The segment $B_0$ is included even though it does not meet $S(55)$;
$B_4$ is the final segment, with $\log16<3$.}
\label{fig:cover-construction}
\end{figure}
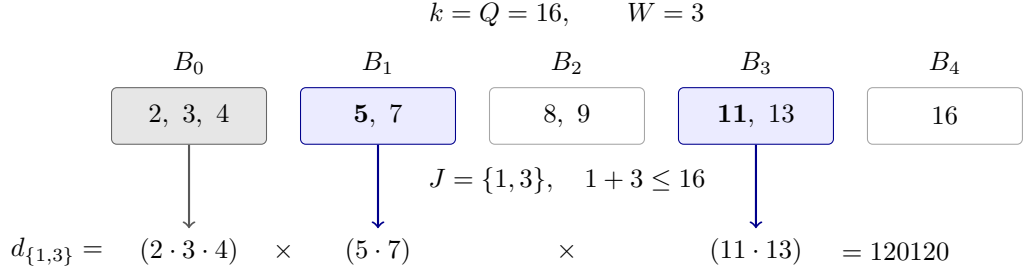

\begin{proposition}\label{prop:elementary-cover}
For integers $k\ge2$ and all sufficiently large integers $Q$, the
construction above gives a cover $\mathcal C_k$ satisfying
\begin{equation}\label{eq:deterministic-tradeoff}
 \begin{aligned}
 |\mathcal C_k|&\le e^{2\sqrt Q},\\
 \max_{d\in\mathcal C_k}\log d
 &\le(1+\sqrt{2Q})\left(\frac{3k}{Q}+\log k\right).
 \end{aligned}
\end{equation}
\end{proposition}
\begin{proof}
Since $\max B_i\le k$, the defining inequalities give, for every segment,
\[
 \sum_{z\in B_i}\log z
 =\sum_{\substack{z\in B_i\\z<\max B_i}}\log z+\log(\max B_i)
 <W+\log k.
\]
For any $z\in B_i$, the $i$ preceding segments have total weight at
least $iW$, and all their elements are smaller than $z$.
Lemma~\ref{lem:elementary-weight} gives
\[
 iW\le\sum_{\substack{s\in\cP_k\\s<z}}\log s
     \le A(z)\le3z.
\]
Thus every element of $B_i$ is at least $iW/3=ik/Q$.

Fix $S\in\cF_k$ and let $J$ be the set of positive indices of segments
that meet $S$:
\[
 J=\{i\ge1:S\cap B_i\ne\varnothing\}.
\]
The sets $S\cap B_i$ for $i\in J$ are nonempty and disjoint, so
\begin{equation}\label{eq:group-index-budget}
 \frac{k}{Q}\sum_{i\in J}i
 \le\sum_{i\in J}\min(S\cap B_i)
 \le\sum_{z\in S}z\le k.
\end{equation}
Hence $\sum_{i\in J}i\le Q$, and this $J$ occurs in the construction.
Every element of $S$ belongs to $B_0$ or to a segment indexed by $J$.
Therefore the product in~\eqref{eq:cover-integer} contains every element
of $S$, giving $\prod_{z\in S}z\mid d_J$.
Lemma~\ref{lem:support} now shows that $\mathcal C_k$ is a cover.

Every $J$ used in the construction is a subset of $\{1,\ldots,Q\}$,
since its positive indices sum to at most $Q$.
Since $\pi/\sqrt3<2$, Lemma~\ref{lem:elementary-count} gives
$|\mathcal C_k|\le e^{2\sqrt Q}$ for all sufficiently large $Q$.
For every $J$ used in the construction,
\[
 \frac{|J|(|J|+1)}2\le\sum_{i\in J}i\le Q,
 \qquad |J|\le\sqrt{2Q}.
\]
Each $d_J$ uses $|J|+1$ segments, each with logarithmic sum less than
$W+\log k$.  Thus
\[
 \log d_J\le(1+\sqrt{2Q})(W+\log k),
\]
which gives~\eqref{eq:deterministic-tradeoff} after substituting $W=3k/Q$.
\end{proof}

The construction enumerates index sets, not permutation orders.
Numerical ordering is essential: it converts a segment's index into a
lower bound on the cost of each element in it.
For an integer-arithmetic implementation, use weights
$\lceil\log_2z\rceil\le2\log z/\log2$.  Lemma~\ref{lem:elementary-weight} gives
\[
 \sum_{p^a\le x}\lceil\log_2(p^a)\rceil
 \le\frac{2A(x)}{\log2}\le9x.
\]
Use the threshold $\lceil9k/Q\rceil$.  Each segment then has weight at most
$\lceil9k/Q\rceil+\lceil\log_2k\rceil$.

\begin{proof}[Proof of Theorem~\ref{thm:elementary-identity}]
Set $Q=\lfloor(\log k)^2/16\rfloor$.  Then $Q\to\infty$ as $k\to\infty$
and $Q=\Theta((\log k)^2)$.  For sufficiently large $k$,
Proposition~\ref{prop:elementary-cover} gives a cover $\mathcal C_k$ with
\[
 \begin{aligned}
 |\mathcal C_k|&\le e^{2\sqrt Q}\le\sqrt k,\\
 \max_{d\in\mathcal C_k}\log d
 &\le(1+\sqrt{2Q})\left(\frac{3k}{Q}+\log k\right)\\
 &=12\sqrt2\,\frac{k}{\log k}
   +O\left(\frac{k}{(\log k)^2}\right).
 \end{aligned}
\]
Lemma~\ref{lem:cover-identity} now gives distinct binary words with equal
letter counts, forming an identity of $\Sym_k$, and
\[
 \begin{aligned}
 \log|u_k|
 &\le\log2+\bigl(|\mathcal C_k|-1\bigr)\log3
       +\max_{d\in\mathcal C_k}\log d\\
 &\le12\sqrt2\,\frac{k}{\log k}
   +O\left(\frac{k}{(\log k)^2}\right).
 \end{aligned}
\]
Exponentiating gives the claimed bound.
\end{proof}

\section{Tightness of the construction}
\label{sec:tightness}

Section~\ref{sec:elementary-cover} constructs a cover with at most
$\sqrt k$ integers and $\max_{d\in\mathcal C_k}\log d=O(k/\log k)$.
The following lemma gives a matching lower bound when $\ell=\sqrt k$
and describes its dependence on the allowed cover size $\ell$ for
$2\le\ell\le k$.

\begin{lemma}[Cover bound]\label{lem:cover-lower}
For $2\le\ell\le k$, every cover $\mathcal C_k$ with
$|\mathcal C_k|\le\ell$ satisfies
\[
 \max_{d\in\mathcal C_k}\log d
 =\Omega\!\left(\frac{k}{\log\ell}\right),
\]
with an absolute implicit constant, uniformly in $\ell$.
\end{lemma}
\begin{proof}
Set $t=\lfloor k/(8\log\ell)\rfloor$, and let $N$ be the number of
primes in $(t,4t]$.  For sufficiently large $k$,
$t\ge k/(16\log\ell)$, uniformly for $2\le\ell\le k$.

The exponent of each prime $p$ in $\binom{4t}{2t}$ is at most
$\lfloor\log(4t)/\log p\rfloor$.  Powers $p^a\le4t$ with $a\ge2$
contribute at most $2\sqrt t\log(4t)$, since their primes are at most
$2\sqrt t$ and each contributes at most $\log(4t)$.
Using $\psi(t)\le2t\log2$ from the proof of
Lemma~\ref{lem:elementary-weight} and the largest binomial coefficient,
\[
 \begin{aligned}
 \sum_{t<p\le4t}\log p
 &\ge\log\binom{4t}{2t}-2t\log2-2\sqrt t\log(4t)\\
 &\ge2t\log2-\log(4t+1)-2\sqrt t\log(4t)\\
 &\ge t\log2
 \end{aligned}
\]
for sufficiently large $t$.  Consequently,
\[
 N\ge\frac{t\log2}{\log(4t)}
 \ge\frac{k\log2}{16(\log k)^2}
 \ge\lceil\log\ell\rceil.
\]
Every product of $\lceil\log\ell\rceil$ distinct primes in $(t,4t]$
is an element order of $\Sym_k$: use those primes as disjoint cycle
lengths and fix the remaining points.  Their sum is at most
$4t\lceil\log\ell\rceil\le k$, since $\lceil\log\ell\rceil\le2\log\ell$.

Suppose $\log d<N\log t/e$ for every $d\in\mathcal C_k$.
Then each $d$ is divisible by fewer than $N/e$ of the $N$ primes.
Removing such prime divisors one at a time cannot increase their
proportion among the remaining primes.  Thus each $d$ covers less
than the fraction $e^{-\lceil\log\ell\rceil}\le1/\ell$ of the above
products.  Since $|\mathcal C_k|\le\ell$, not all products would be
covered, a contradiction.  Therefore
\[
 \begin{aligned}
 \max_{d\in\mathcal C_k}\log d
 &\ge\frac Ne\log t\\
 &\ge\frac{\log2}{e}\,t\frac{\log t}{\log(4t)}
 \ge\frac{\log2}{32e}\,\frac{k}{\log\ell}.
 \end{aligned}
\]
The last inequality uses $\log t/\log(4t)\ge1/2$ for sufficiently
large $k$.
\end{proof}

\medskip

With reference to the particular words constructed in Lemma~\ref{lem:cover-identity},
for any cover $\mathcal C_k=\{d_1,\ldots,d_r\}$, the above lower bound has the following effect:
Each block exponent in $U_i$, added to its replacement in
$C_i=U_i^{\dagger d_{i+1}}$, is a positive multiple of $d_{i+1}$.
Thus $|U_i|+|C_i|\ge d_{i+1}$, and the recursion gives
\[
 |U_{i+1}|=2|U_i|+|C_i|
 \ge\max\{2|U_i|,d_{i+1}\}.
\]
Since $|U_1|=d_1\ge1$, induction and the final substitution give
\begin{equation}\label{eq:construction-length-lower}
 |u|=|v|=2|U_r|
 \ge\max\{2^r,2\max\mathcal C_k\}.
\end{equation}
If $r>k$, then $\log|u|\ge r\log2>k\log2$.
If $r\le k$, Lemma~\ref{lem:cover-lower} with $\ell=k$ gives
$\log|u|\ge\max_{d\in\mathcal C_k}\log d=\Omega(k/\log k)$.
Together with Theorem~\ref{thm:elementary-identity}, this shows that the
minimum logarithmic word length attainable by this construction is
$\Theta(k/\log k)$.

Thus, changing the cover or the order of enumerating $d_1,\dots,d_r$ cannot improve the order
$k/\log k$ for this recursion.  Constants and lower-order terms can
still be improved, as in Appendix~\ref{sec:refinement}.  Furthermore, the facts that the numerator of $\frac{k}{\log \ell}$ is independent of $\ell$, and that the denominator is in $\log\ell$, impose severe limitations on possible improvements via this strategy.  We next explore whether other
classes of identities (considered in \cite{KS}) can reach the Landau scale and give better bounds on $\Sepp(n)$.

\subsection{Frontier Tradeoffs}

We compare the bound just proved with the maximum element order and
the common exponent of $\Sym_k$.
Let $M(k)$ be \emph{Landau's function} and $E(k)$ the
\emph{exponent of the symmetric group}:
\[
 \begin{aligned}
 M(k)&:=\max_{\sigma\in\Sym_k}\ord(\sigma),\\
 E(k)&:=\min\{N\ge1:\sigma^N=1\text{ for every }\sigma\in\Sym_k\}.
 \end{aligned}
\]
A permutation with cycle lengths $c_1,\ldots,c_t$ has order
$\lcm(c_1,\ldots,c_t)$.  Thus $M(k)$ is the maximum of this least
common multiple subject to $\sum_i c_i=k$.  The common exponent is
\[
 E(k)=\lcm(1,2,\ldots,k)
     =\prod_{p\le k}p^{\lfloor\log k/\log p\rfloor},
\]
since every cycle length divides this number and $\Sym_k$ contains a
$j$-cycle for each $1\le j\le k$.  For example, $M(5)=6$ and $E(5)=60$;
the maximum order is attained by disjoint cycles of lengths $2$ and $3$.
Landau's theorem (equation~(6) in~\cite{DNZ}) and the prime number
theorem~\cite{DiamondSteinig}
give, respectively,
\[
 M(k)=\exp\!\bigl((1+o(1))\sqrt{k\log k}\bigr),\qquad
 E(k)=\exp\!\bigl((1+o(1))k\bigr).
\]
Both estimates are unconditional; for the second,
$\log E(k)=\sum_{p^a\le k}\log p\sim k$.

The identity $0^{E(k)}=1^{E(k)}$ yields the lower bound
$\Sepp(n)=\Omega(\log n)$.
Proposition~14 in~\cite{KS} gives a length lower bound of $M(k)/2$
for numerical identities.  Theorem~\ref{thm:elementary-identity}
reduces the previous upper bound $\exp(2k/3+O(k/\log k))$ to
$\exp(O(k/\log k))$ within this class.
If, for every sufficiently large $k$, we
could construct distinct equal-length positive binary identities with
\[
 |u_k|=|v_k|\le M(k)^{O(1)}
                 =\exp\!\bigl(O(\sqrt{k\log k})\bigr),
\]
the inversion and padding argument in Section~\ref{sec:introduction}
would give
\[
 \Sepp(n)\ge\Sepp_{=}(n)
   =\Omega\!\left(\frac{(\log n)^2}{\log\log n}\right).
\]
Conjecture~15 in~\cite{KS} states that a shortest
positive identity of $\Sym_k$ is numerical for all sufficiently large $k$.
Corollary~16 in~\cite{KS} then gives $\Sepp(n)=O((\log n)^2/\log\log n)$.
The hypothetical construction above would match this conditional upper
bound in order of magnitude.

Every permutation order is at most $M(k)$, but need not divide $M(k)$.
The set $\{1,\ldots,M(k)\}$ therefore covers every permutation order;
$\{M(k)\}$ need not.  Lemma~\ref{lem:cover-identity} gives
\[
 |u_k|=|v_k|\le2M(k)3^{M(k)-1},
\]
an upper bound larger than $E(k)$ for large $k$.  A positive-word
construction whose length is polynomial in the size of the cover and
its largest integer would attain the $M(k)^{O(1)}$ target with this cover.

When inverse letters are allowed, Kozma and Thom's combination
lemma (Lemma~2.2 in~\cite{KozmaThom}) gives a nontrivial group law of length
at most $16M(k)^3$.  Their stronger construction (Theorem~1.3
in~\cite{KozmaThom}) gives length $\exp(O((\log k)^4\log\log k))$.
These constructions use nested commutators and do not directly yield
positive identities of comparable length (see Section~1 in~\cite{BKSS}).
Replacing $\sigma_i^{-1}$ by
$\sigma_i^{E(k)-1}$, for $i\in\{0,1\}$, can reintroduce the cost of
$E(k)$.  The remaining challenge is to combine the possible exponents
into a short positive binary identity.
For the recursion~\eqref{eq:selector-recursion}, however,
Lemma~\ref{lem:cover-lower} shows that no choice or enumeration of the cover
can give $\log|u_k|=o(k/\log k)$.

\section{AI Disclosure and Concluding Discussion}
\label{sec:disclosure}

The $\Omega(\log n \log\log n)$ lower bound came from usage of GPT 5.6 Sol (Pro level) in July 2026.  Subsequent effort toward stronger lower bounds improved the leading
coefficient to $1-o(1)$, as reflected in Appendix~\ref{sec:refinement}.  
The original AI proof of the cover bounds in Section~\ref{sec:elementary-cover} used the classic probabilistic method.  This proof was humanly compacted into 8 pages.  The AI was then prompted to replace the probabilistic method by a deterministic explicit construction, which (after some further polishing by both human and AI) produced the proof presented here.  Further effort was devoted to mapping out the frontier, in particular:

\begin{itemize}
    \item[$\bullet$] determining which improvements to the word identity bounds in Section~\ref{sec:selector} would give order $\log^2 n/\log\log n$ per Karpova and Shur \cite{KS}.
    \item[$\bullet$]
    in-between, which might yield order $(\log n)^{1+\epsilon}(1 - o(1))$ for $0 < \epsilon < 1$.
    \item[$\bullet$] 
    whether any lower bounds of polynomial order $n^\epsilon$ might even be possible using this approach.
\end{itemize}

After trying several more heuristic ideas, GPT 5.6 opined that any improvement of the lower bound to $\Omega((\log n)^{1+\epsilon})$ would require identity size exponential in $(\log n)^{1/\epsilon}$, i.e., bringing things all the way down to quasi-polynomial size.  Doing so strikes us as embodying much of the original difficulty of the permutation separating words problem to begin with.  We therefore believe that we have probed the frontier enough to aver that substantial immediate improvements to the scheme of this paper will require new ideas that did not drop out from further AI-assisted trials.


\clearpage
\appendix
\section{Optimized constants and lower-order terms}
\label{sec:refinement}

We sharpen the bounds by a greedy construction.  The prime
estimates use the quantitative prime number theorem of Diamond
and Steinig~\cite{DiamondSteinig}.  Put
\begin{equation}\label{eq:notation}
 L=\log k,\qquad \ell=\log L,\qquad
 h=\log(\log3),\qquad \kappa=1+h,\qquad c_0=1-h.
\end{equation}
\begin{theorem}\label{thm:identity}
As $k\to\infty$,
\begin{equation}\label{eq:main-identity}
 \log\ell_{+,2}(k)
 \le \frac{k}{L}
   +\bigl(3\ell+\log\ell+\kappa+o(1)\bigr)\frac{k}{L^2}.
\end{equation}
The constructed binary words have the same number of occurrences of each
letter.
\end{theorem}

\begin{theorem}\label{thm:sepp}
As $n\to\infty$,
\begin{equation}\label{eq:main-sepp}
 \begin{split}
 \Sepp_{=}(n)\ge\log n\bigl(&\log\log n-2\log\log\log n\\
 &-\log\log\log\log n-\kappa-o(1)\bigr).
 \end{split}
\end{equation}
The same lower bound holds for $\Sepp(n)$.
\end{theorem}

The number-theoretic estimate we need is
\begin{equation}\label{eq:pnt}
 \theta(z)=z+O_A\left(\frac{z}{(\log z)^A}\right),\qquad
 \psi(z)=z+O_A\left(\frac{z}{(\log z)^A}\right)
\end{equation}
for every fixed $A>0$, where
$\theta(z)=\sum_{p\le z}\log p$ and
$\psi(z)=\sum_{p^a\le z}\log p$.
Diamond and Steinig~\cite{DiamondSteinig} give an elementary proof of
\[
 \psi(z)=z+O\!\left(z\exp\!\left\{
 -\frac{(\log z)^{1/7}}{(\log\log z)^2}\right\}\right),
\]
whose error is $o(z/(\log z)^A)$ for every fixed $A>0$.
Since $\psi(z)-\theta(z)=O(\sqrt z\log z)$ by
Lemma~\ref{lem:elementary-weight}, this proves~\eqref{eq:pnt}.
In particular,
$\log\lcm(1,\ldots,\lfloor z\rfloor)=\psi(z)$.
The distinct weight sum $\sum_{p^a\le z}\log(p^a)$ is not $\psi(z)$.
When needed below, its proper-prime-power part is bounded separately by
$O(\sqrt z(\log z)^2)$.

\subsection{Interval products}

Use the notation in~\eqref{eq:notation}, and choose
\begin{equation}\label{eq:parameters}
 y=\frac{k}{L^2},\qquad
 \delta=\frac{k}{L^3\ell^4},\qquad
 x=L-3\ell-\log\ell+c_0,\qquad \lambda=\frac{x}{k}.
\end{equation}
Put $P_y=\lcm(1,\ldots,\lfloor y\rfloor)$.
Partition $(y,k]$ into intervals
\[
 I_i=(z_i,\min(z_i+\delta,k)],\qquad
 z_i=y+(i-1)\delta,
 \quad 1\le i\le m:=\left\lceil\frac{k-y}{\delta}\right\rceil.
\]
Set
\[
 W_i=\prod_{s\in\cP_k\cap I_i}s,\qquad w_i=\log W_i,
 \qquad q_i=e^{-\lambda(z_i-y)}.
\]
For every $H\subseteq\{1,\ldots,m\}$, define
\begin{equation}\label{eq:interval-products}
 D_H=P_y\prod_{i\in H}W_i,\qquad
 a_H=\prod_{i\in H}q_i\prod_{i\notin H}(1-q_i).
\end{equation}
Empty products equal one.  Expanding the finite product gives
\[
 \sum_Ha_H=\prod_{i=1}^m\bigl(q_i+(1-q_i)\bigr)=1.
\]

\subsubsection{Counting index sets}

For an element order $q$ of $\Sym_k$, let $J$ be the indices of the
intervals containing elements of $S(q)$ greater than $y$.
By choosing one element of $S(q)$ from each occupied interval and applying
Lemma~\ref{lem:support},
\begin{equation}\label{eq:feasible}
 \sum_{i\in J}z_i\le k,\qquad j:=|J|\le L^2.
\end{equation}
If $J\subseteq H$, then $q\mid D_H$.
Call any interval set satisfying~\eqref{eq:feasible} feasible, including
the empty set.
For any fixed feasible $J$ of size $j$,
\begin{equation}\label{eq:containment}
 \sum_{H\supseteq J}a_H
 =\prod_{i\in J}q_i
 =e^{-\lambda\sum_{i\in J}(z_i-y)}
 \ge p_j:=e^{-x(1-j/L^2)}.
\end{equation}
The same lower bound holds for $J=\varnothing$.

Put
\[
 K=\lceil k/\delta\rceil+\lfloor L^2\rfloor
   =(1+o(1))L^3\ell^4.
\]
Every feasible set satisfies $\sum_{i\in J}i\le K$.  If $M_j$ denotes
the number of feasible sets of cardinality $j\ge1$, then
\begin{equation}\label{eq:interval-count}
 M_j\le\frac{\binom Kj}{j!}
 \le\frac{K^j}{(j!)^2}
 \le\left(\frac{e^2K}{j^2}\right)^j.
\end{equation}
Indeed, each distinct $j$-set has $j!$ orderings, all positive integer
$j$-tuples of sum at most $K$; there are exactly $\binom Kj$ such tuples.

For $j\ge1$ with $M_j>0$, define
\[
 t=\frac{xj}{L^2},\qquad
 C_k=2+\log K-4\log L+2\log x
    =\ell+4\log\ell+2+o(1).
\]
Combining~\eqref{eq:containment} and~\eqref{eq:interval-count} gives
\begin{equation}\label{eq:ratio}
 p_j^{-1}\log M_j
 \le e^x\frac{L^2}{x}\,te^{-t}(C_k-2\log t)
 \le e^x\frac{L^2}{x}\frac{C_k+2}{e}.
\end{equation}
For the last inequality, $te^{-t}\le1/e$; when $t\ge1$ the term
$-2te^{-t}\log t$ is nonpositive, and when $0<t<1$ it is at most $2/e$.
Thus the maximum in~\eqref{eq:ratio} is bounded by
\begin{equation}\label{eq:count-scale}
 \left(\frac1e+o(1)\right)e^xL\ell.
\end{equation}
The reduction to interval sets, rather than individual supports,
is what gives this scale.

\subsubsection{Bounding the products}

\begin{lemma}\label{lem:interval-weight}
Uniformly over the intervals,
\begin{equation}\label{eq:interval-weight}
 w_i\le(1+\varepsilon_k)\delta,\qquad
 \varepsilon_k=o\left(\frac1{L\ell}\right).
\end{equation}
\end{lemma}
\begin{proof}
For primes, subtract the endpoint estimates in~\eqref{eq:pnt}, using
$A=12$.  Both endpoints are at least $k/L^2$, so the error is
$O(k/L^{12})$, uniformly.  The total contribution of all proper prime
powers at most $k$ is $O(\sqrt{k}L^2)$; this bounds their contribution
to any one interval.  Both errors are $o(\delta/(L\ell))$.
The last interval may be shorter, which does not affect the upper bound.
\end{proof}

Put $X_H=\log D_H$ and
$\mu=\sum_Ha_HX_H=\psi(y)+\sum_iq_iw_i$.
The geometric series and Lemma~\ref{lem:interval-weight} give
\begin{align}
 \mu
 &\le\psi(y)+(1+\varepsilon_k)
       \frac{\delta}{1-e^{-\lambda\delta}}\notag\\
 &\le\frac{k}{x}+\frac{k}{L^2}
       +o\left(\frac{k}{L^2\ell}\right).
 \label{eq:weighted-size}
\end{align}
Here $\lambda\delta\to0$,
$\delta/(1-e^{-\lambda\delta})=1/\lambda+O(\delta)$, and
$\psi(y)=y+o(k/(L^2\ell))$ by~\eqref{eq:pnt}.

For a feasible $J$, define
\[
 \mu_J=\psi(y)+\sum_{i\in J}w_i+\sum_{i\notin J}q_iw_i
      =\mu+\sum_{i\in J}(1-q_i)w_i.
\]
Using $1-e^{-u}\le u$ and~\eqref{eq:feasible},
\begin{equation}\label{eq:forced-size}
 \mu_J-\mu
 \le(1+\varepsilon_k)\delta\lambda\sum_{i\in J}(z_i-y)
 \le(1+\varepsilon_k)\delta x
 =O\left(\frac{k}{L^2\ell^4}\right).
\end{equation}
Expanding over $H\setminus J$ gives the quadratic identity
\begin{align}
 \frac{\sum_{H\supseteq J}a_H(X_H-\mu_J)^2}
      {\prod_{i\in J}q_i}
 &=\sum_{i\notin J}w_i^2q_i(1-q_i)\notag\\
 &\le(\max_iw_i)\sum_iw_iq_i
 =O(\delta k/L)
 =O\left(\frac{k^2}{L^4\ell^4}\right).
 \label{eq:quadratic-sum}
\end{align}
Set
\begin{equation}\label{eq:cap}
 B_k=\frac{k}{x}+\frac{k}{L^2}+\frac{k}{L^2\ell},
 \qquad \gamma_k=1-\frac1\ell.
\end{equation}
By~\eqref{eq:weighted-size} and~\eqref{eq:forced-size},
$B_k-\mu_J\ge k/(2L^2\ell)$ for all sufficiently large $k$,
uniformly over feasible $J$.
For $X_H>B_k$ we have $(X_H-\mu_J)^2\ge(B_k-\mu_J)^2$, so
\[
 \frac{\sum_{\substack{H\supseteq J\\X_H>B_k}}a_H}
      {\prod_{i\in J}q_i}
 \le
 \frac{\sum_{i\notin J}w_i^2q_i(1-q_i)}{(B_k-\mu_J)^2}
 =O(\ell^{-2})\le\frac1\ell.
\]
Thus~\eqref{eq:containment} gives
\begin{equation}\label{eq:good}
 \sum_{\substack{H\supseteq J\\X_H\le B_k}}a_H
 \ge\gamma_k\prod_{i\in J}q_i\ge\gamma_kp_j.
\end{equation}

\subsection{Greedy covering}

Put
\begin{equation}\label{eq:cover-count}
 R=\left\lceil\frac{e^x}{\gamma_k}
 \left(\frac{L^2}{x}\frac{C_k+2}{e}+4\ell\right)\right\rceil.
\end{equation}
Since $e^xp_j\ge1$, \eqref{eq:ratio} and~\eqref{eq:cover-count} imply
\[
 R\gamma_kp_j\ge\log M_j+4\ell
\]
for every cardinality with $M_j>0$, including $j=0$ with $M_0=1$.
Let $\mathcal U_t$ be the feasible sets not covered after $t$ choices, and set
\[
 \Phi_t=\sum_{J\in\mathcal U_t}
             e^{-(R-t)\gamma_kp_{|J|}}.
\]
Initially,
\[
 \Phi_0=\sum_{j=0}^{\lfloor L^2\rfloor}M_j e^{-R\gamma_kp_j}
 \le(\lfloor L^2\rfloor+1)e^{-4\ell}<1.
\]
At step $t+1$, choose $H$ with $X_H\le B_k$ maximizing
\[
 \sum_{\substack{J\in\mathcal U_t\\J\subseteq H}}
        e^{-(R-t-1)\gamma_kp_{|J|}},
\]
breaking ties lexicographically, and remove the covered sets from
$\mathcal U_t$.  Since $\sum_{H:X_H\le B_k}a_H\le1$, the maximum
is at least its weighted sum over such $H$.  By~\eqref{eq:good},
that sum is at least
\[
 \sum_{J\in\mathcal U_t}\gamma_kp_{|J|}
              e^{-(R-t-1)\gamma_kp_{|J|}}.
\]
Using $1-u\le e^{-u}$ gives
\[
 \Phi_{t+1}
 \le\sum_{J\in\mathcal U_t}(1-\gamma_kp_{|J|})
             e^{-(R-t-1)\gamma_kp_{|J|}}
 \le\Phi_t.
\]
Hence $|\mathcal U_R|=\Phi_R\le\Phi_0<1$, so $\mathcal U_R=\varnothing$.
The chosen integers $D_H$ form a cover $\mathcal C_k$, with
$|\mathcal C_k|\le R$ and
$\max_{d\in\mathcal C_k}\log d\le B_k$.

The number of choices satisfies
\begin{align}
 R
 &\le\left(\frac1e+o(1)\right)e^xL\ell\notag\\
 &=\left(e^{c_0-1}+o(1)\right)\frac{k}{L^2}
  =\left(\frac1{\log3}+o(1)\right)\frac{k}{L^2}.
 \label{eq:cover-size}
\end{align}
In deriving this, $C_k/\ell\to1$, $x/L\to1$, and the $4\ell$ slack in
\eqref{eq:cover-count} is lower order.  The bound
$\max_{d\in\mathcal C_k}\log d\le B_k$ uses the expansion
\begin{align}
 B_k
 &=\frac{k}{L}
   +\bigl(3\ell+\log\ell-c_0+1+o(1)\bigr)\frac{k}{L^2}\notag\\
 &=\frac{k}{L}
   +\bigl(3\ell+\log\ell+h+o(1)\bigr)\frac{k}{L^2}.
 \label{eq:integer-size}
\end{align}
The next Taylor term is $O(k\ell^2/L^3)=o(k/L^2)$.
Combining~\eqref{eq:cover-size}, \eqref{eq:integer-size}, and
Lemma~\ref{lem:cover-identity} proves Theorem~\ref{thm:identity}.

\subsection{Exact-length permutation separation}
\label{sec:inversion}

We give the inversion with its lower-order terms to prove
Theorem~\ref{thm:sepp}.  Put
\[
 T=\log n,\qquad A=\log T,\qquad B=\log A,
\]
and choose
\begin{equation}\label{eq:k-choice}
 k=\left\lfloor
 T\bigl(A-2B-\log B-\kappa-\varepsilon_n\bigr)
 \right\rfloor,
\end{equation}
where $\varepsilon_n>0$ tends to zero sufficiently slowly.
Writing $s=-2B-\log B-\kappa-\varepsilon_n$, we have
\[
 \begin{split}
 \log k&=A+B+O(B/A),\\
 \log\log k&=B+O(B/A),\\
 \log\log\log k&=\log B+O(1/A).
 \end{split}
\]
In particular, the upper bound in~\eqref{eq:main-identity}, divided by
$T$, is
\begin{align*}
 1+\frac{s-B}{A}
  +\frac{3B+\log B+\kappa}{A}+o(1/A)
 &=1+\frac{s+2B+\log B+\kappa+o(1)}{A}\\
 &=1+\frac{-\varepsilon_n+o(1)}{A}<1.
\end{align*}
To make the last inequality precise, choose $\varepsilon_n$ to dominate
both $B^2/A$ and a tail envelope of the $o(1)$ remainder in
Theorem~\ref{thm:identity}; all candidate $k$ lie in
$[TA/2,2TA]$ for large $n$.  Such an envelope tends to zero, so this
choice is possible with $\varepsilon_n\to0$.

Thus the constructed words $u_k,v_k$ have common length less than $n$.
The distinct words $u_k0^{n-|u_k|}$ and $v_k0^{n-|u_k|}$ have length
exactly $n$ and form an identity of $\Sym_k$.  By
the correspondence in Section~\ref{sec:introduction}, they cannot be separated
by a permutation automaton with at most $k$ states.  Substituting~\eqref{eq:k-choice} proves
\eqref{eq:main-sepp}; $\Sepp(n)\ge\Sepp_{=}(n)$ gives the second assertion.

\end{document}